\documentclass{llncs}

\usepackage{amssymb,amsmath,graphicx}
\usepackage{color,enumerate,boxedminipage,tkz-graph}
\usepackage{pgf,float,hyperref}
\usepackage{subcaption}
\usepackage[T1]{fontenc}

\newcommand{\NP}{{\sf NP}}
\newcommand{\sP}{{\sf P}}

\newcommand{\problemdef}[3]{
	\begin{center}
		\begin{boxedminipage}{.99\textwidth}
			\textsc{{#1}}\\[2pt]
			\begin{tabular}{ r p{0.8\textwidth}}
				\textit{~~~~Instance:} & {#2}\\
				\textit{Question:} & {#3}
			\end{tabular}
		\end{boxedminipage}
	\end{center}
}

\newtheorem{open}[example]{Open Problem}
\newtheorem{observation}{Observation}

\title{Steiner Trees for Hereditary Graph Classes:\\ a Treewidth Perspective\thanks{An extended abstract of this paper appeared in the proceedings of LATIN 2020~\cite{BBJPP20}. The paper received support from the Leverhulme Trust (RPG-2016-258) and the 
Royal Society (IES$\setminus$R1$\setminus$191223).}} 

\author{
Hans Bodlaender\inst{1}
\orcidID{0000-0002-9297-3330}
\and
Nick Brettell\inst{2}
\orcidID{0000-0002-1136-418X}
\and
Matthew Johnson\inst{2}
\orcidID{0000-0002-7295-2663}
\and
Giacomo Paesani\inst{2}
\orcidID{0000-0002-2383-1339}
\and
Dani\"el Paulusma\inst{2}
\orcidID{0000-0001-5945-9287}
\and 
Erik Jan van Leeuwen\inst{1}}

\institute{Department of Information and Computing Sciences,
Utrecht University, The Netherlands \email{\{h.l.bodlaender,e.j.vanleeuwen\}@uu.nl} \and
Department of Computer Science, Durham University, UK \\ \email{\{nicholas.j.brettell,matthew.johnson2,giacomo.paesani,daniel.paulusma\}@durham.ac.uk}}

\begin{document}
\maketitle 

\begin{abstract}
We consider the classical problems {\sc (Edge) Steiner Tree} and {\sc Vertex Steiner Tree} after restricting the input to some class of graphs characterized by a small set of forbidden induced subgraphs. We show a dichotomy for the former problem restricted to $(H_1,H_2)$-free graphs and a dichotomy for the latter problem restricted to $H$-free graphs. We find that there exists an infinite family of graphs $H$ such that {\sc Vertex Steiner Tree} is polynomial-time solvable for $H$-free graphs, whereas there exist only two graphs $H$ for which this holds for {\sc Edge Steiner Tree}
(assuming $\sP\neq\NP$). 
We also find that {\sc Edge Steiner Tree} is polynomial-time solvable for $(H_1,H_2)$-free graphs if and only if the treewidth of the class of $(H_1,H_2)$-free graphs is bounded (subject to $\sP\neq\NP$).
To obtain the latter result, we determine all pairs $(H_1,H_2)$ for which the class of $(H_1,H_2)$-free graphs has bounded treewidth.
\end{abstract}


\section{Introduction}

Let $G=(V,E)$ be a connected graph and $U\subseteq V$ be a set of \emph{terminal} vertices.
 A {\it Steiner tree} for $U$ (of $G$) is a tree in $G$ that contains all vertices of~$U$. An \emph{edge weighting} of $G$ is a function $w_E: E\rightarrow \mathbb{Q}^+$ (where $\mathbb{Q}^+$ denotes the set of strictly positive rational numbers).
For a tree $T$ in $G$, the {\it edge weight~$w_E(T)$} of $T$ is the sum $\sum_{e \in E(T)}w_E(e)$.
We consider the classical problem:

\problemdef{{\sc Edge Steiner Tree}}{a connected graph $G=(V,E)$ with an  weighting $w_E$, a subset $U\subseteq V$ of terminals and a positive integer $k$.}{does~$G$ have a Steiner tree $T_U$ for $U$ with $w_E(T_U)\leq k$?}

This is often known simply as {\sc Steiner Tree}, but we wish to distinguish it from a closely related problem.  
A \emph{vertex weighting} of $G$ is a function $w_V: V\rightarrow \mathbb{Q}^+$.
For a tree $T$ in $G$, the {\it vertex weight~$w_V(T)$} of $T$ is the sum 
 $\sum_{v\in V(T)}w(v)$.  The following problem is sometimes known as {\sc Node-Weighted Steiner Tree}.

\problemdef{{\sc Vertex Steiner Tree}}{a connected graph $G=(V,E)$ with a vertex weighting $w_V$, a subset $U\subseteq V$ and a positive integer~$k$.}{does~$G$ have a Steiner tree $T_U$ for $U$ with $w_V(T_U)\leq k$?}

Note that we defined the problem for connected inputs, as for general inputs a Steiner tree exists only if the vertices of $U$ all belong to the same connected component.
Note also that {\sc Edge Steiner Tree} 
is a generalization of the {\sc Spanning Tree} problem (set $U=V(G)$). 
We refer to the textbooks of Du and Hu~\cite{DH08} and Pr\"omel and Steger~\cite{PS12} for further background information on Steiner trees. 

We consider the problems {\sc Edge Steiner Tree} and {\sc Vertex Steiner Tree} separately so that, for any graph under consideration, we have either an edge or vertex weighting but not both, so we will generally denote weightings by $w$ without any subscript.  Moreover, when we use the following terminology there is no ambiguity.
We say that a Steiner tree of least possible weight is \emph{minimum}, and that an instance of one of the two problems is \emph{unweighted} if the weighting is constant.
We denote instances of the weighted problems by $(G,w,U,k)$ and of the unweighted problems by $(G,U,k)$.
It is well known that the unweighted versions of {\sc Edge Steiner Tree} and {\sc Vertex Steiner Tree} are \NP-complete~\cite{Ka72,FPW85}.
Moreover, as an $n$-vertex tree has exactly $n-1$ edges, one can make the following observation.

\begin{observation}\label{o-equivalent}
The unweighted versions of {\sc Edge Steiner Tree} and {\sc Vertex Steiner Tree} are equivalent.
\end{observation}

\noindent
{\bf Our Focus}
We focus on the computational
complexity of {\sc Edge Steiner Tree} and {\sc Vertex Steiner Tree} for {\it hereditary} graph classes, 
which are those
graph classes that are 
closed under vertex deletion.
We do this from a {\it systematic} point of view. 
It is well known, and readily seen, that a graph class ${\cal G}$ is hereditary if and only if it can be characterized by a set ${\cal H}$ of forbidden induced subgraphs. 
That is, a graph $G$ belongs to ${\cal G}$ if and only if $G$ has no induced subgraph isomorphic to some graph in ${\cal H}$. 
We normally require ${\cal H}$ to be minimal, in which case it is unique and we denote it by ${\cal H}_{\cal G}$.
We note that ${\cal H}_{\cal G}$ may have infinite size; for example, if ${\cal G}$ is the class of bipartite graphs, then ${\cal H}_{\cal G}=\{C_3,C_5,\ldots,\}$, where $C_r$ denotes the cycle on $r$ vertices.
For a systematic complexity study of a graph problem, we may first consider
 {\it monogenic graph classes} or {\it bigenic} graph classes, which are classes ${\cal G}$ with $|{\cal H}_{\cal G}|=1$ or $|{\cal H}_{\cal G}|=2$, respectively. This is the approach we follow here.

\medskip
\noindent
{\bf Our Results}
We prove a dichotomy for {\sc Edge Steiner Tree} for bigenic graph classes in Section~\ref{s-theproof0} and a dichotomy for {\sc Vertex Steiner Tree} for monogenic graph classes in Section~\ref{s-theproof1}.
We relate our first dichotomy to a dichotomy for boundedness of {\it treewidth} of bigenic graph classes.
The parameter treewidth, which we formally define in Section~\ref{s-treewidth}, measures how close a graph is to being a tree. If we wish to solve a problem on some graph class~${\cal G}$ and we know that~${\cal G}$ has small treewidth, then we can try to mimic efficient algorithms for trees to obtain efficient algorithms for ${\cal G}$.
Many discrete optimization problems can be solved in polynomial time on every graph class of bounded treewidth. 
The {\sc Edge Steiner Tree} problem is an example of such a problem (see, for instance,~\cite{CMZ12} or, for a faster algorithm~\cite{BCKN15}).

In order to describe our dichotomies we need to introduce some extra terminology.
We denote the {\it disjoint union} of two vertex-disjoint graphs $G$ and $H$ by $G+H=(V(G)\cup V(H), E(G)\cup E(H))$, and
the disjoint union of $s$ copies of $G$ by~$sG$.
A {\it linear forest} is a disjoint union of paths.
For a graph $H$, a graph is {\it $H$-free} if it has no induced subgraph isomorphic to $H$.
For a set of graphs $\{H_1,\ldots,H_p\}$, a graph is {\it $(H_1,\ldots,H_p)$-free} if it is $H_i$-free for every $i\in \{1,\ldots,p\}$.
We let $K_r$ and $P_r$ denote the complete graph and path on $r$ vertices, respectively.
The {\it complete bipartite} graph $K_{s,t}$ is the graph whose vertex set can be partitioned into two sets $S$ and $T$ of size $s$ and $t$,  respectively,
such that for any two distinct vertices $u,v$, we have $uv\in E$ if and only if $u\in S$ and $v\in T$ or vice versa. 
We call $K_{1,3}$ the {\it claw}.

\begin{theorem}\label{t-main0}
Let $H_1$ and $H_2$ be two graphs. If one of the following cases holds:
\begin{enumerate}
\item $H_1=K_r$ or $H_2=K_r$ for some $r\in \{1,2\}$;
\item $H_1=K_3$ and $H_2=K_{1,3}$, or vice versa;
\item $H_1=K_r$ for some $r\geq 3$ and $H_2=P_3$, or vice versa; or
\item $H_1=K_r$ for some $r\geq 3$ and $H_2=sP_1$ for some $s\geq 1$, or vice versa.\\[-15pt]
\end{enumerate}
then the treewidth of the class of $(H_1,H_2)$-free graphs is bounded and 
{\sc Edge Steiner Tree} is polynomial-time solvable for $(H_1,H_2)$-free graphs. In all other cases,
the treewidth of $(H_1,H_2)$-free graphs is unbounded and {\sc Edge Steiner Tree} is \NP-complete for $(H_1,H_2)$-free graphs.
\end{theorem}

\begin{theorem}\label{t-main}
Let $H$ be a graph. If $H$ is an induced subgraph of $sP_1+P_4$ for some $s\geq 0$, then {\sc Vertex Steiner Tree} is polynomial-time solvable for $H$-free graphs, otherwise even unweighted {\sc Vertex Steiner Tree} is \NP-complete.
\end{theorem}

\noindent
We make the following observations about these two results:

\medskip
\noindent
{\bf 1.} 
In Theorem~\ref{t-main0} we show that
{\sc Edge Steiner Tree} can be solved in polynomial time for $(H_1,H_2)$-free graphs if and only if the treewidth of the class of $(H_1,H_2)$-free graphs is bounded (assuming $\sP\neq\NP$). However, such a 1-to-1 correspondence holds neither between {\sc Vertex Steiner Tree} and treewidth, nor between {\sc Vertex Steiner Tree} and the less restrictive width parameter mim-width. This can be seen as follows.
It is known that {\sc Vertex Steiner Tree} is polynomial-time solvable for a graph class of bounded mim-width provided that
a branch decomposition of constant mim-width can be found in polynomial time for the class~\cite{BK19}.
In particular, every graph class of bounded treewidth has this property.
However, complete graphs, and hence $P_4$-free graphs, have unbounded treewidth, whereas co-bipartite graphs, and hence $3P_1$-free graphs, have unbounded mim-width~\cite{BCM15}.
In Section~\ref{s-con} we discuss this connection between {\sc Edge Steiner Tree} and treewidth further. 
In the same section we also pose a number of open problems.

\medskip
\noindent
{\bf 2.}
Theorem~\ref{t-main0} also provides a dichotomy for boundedness of treewidth of $(H_1,H_2)$-free graphs. For the less restrictive width parameter clique-width 
(or equivalently boolean-width, rank-width, module-width, or NLC-width~\cite{BTV11,Jo98,OS06,Ra08}) or the even less restrictive width parameter mim-width,
such dichotomies have not yet been established for $(H_1,H_2)$-free graphs. We refer to~\cite{DJP19} and~\cite{BHMPP} for state-of-the-art summaries for clique-width and mim-width, respectively.

\medskip
\noindent
{\bf 3.}
The restriction of Theorem~\ref{t-main0} to monogenic graph classes, that is, taking $H=H_1=H_2$, yields only two (trivial) graphs $H$, namely $H=K_1$ or $H=K_2$, for which the restriction of {\sc Edge Steiner Tree} to $H$-free graphs can be solved in polynomial time. In contrast, by Theorem~\ref{t-main}, {\sc Vertex Steiner Tree} can, when restricted to $H$-free graphs, be solved in polynomial time for an infinite family of linear forests~$H$, namely $H=sP_1+P_4$ $(s\geq 0$).

\medskip
\noindent
{\bf 4.}
Theorem~\ref{t-main} is also a dichotomy for the unweighted {\sc Vertex Steiner Tree} problem.
Moreover, as the unweighted versions of {\sc Edge Steiner Tree} and {\sc Vertex Steiner Tree} are equivalent by Observation~\ref{o-equivalent}, Theorem~\ref{t-main} is also a classification of the unweighted version of {\sc Edge Steiner Tree}.

\section{Proof of Theorem~\ref{t-main0}}\label{s-theproof0}

In this section we give a proof for our first dichotomy, which is for {\sc Edge Steiner Tree} for $(H_1,H_2)$-free graphs. 
We note that this is not the first systematic study of {\sc Edge Steiner Tree}. 
For example, Renjitha and Sadagopan~\cite{RS18} proved that  unweighted {\sc Edge Steiner Tree} is \NP-complete for $K_{1,5}$-free split graphs, but can be solved in polynomial time for $K_{1,4}$-free split graphs. 
We present a number of other results from the literature, which we collect in Section~\ref{s-pre}, together with some lemmas that follow from these results. Then in Section~\ref{s-treewidth} we discuss the notion of treewidth; as mentioned, this notion will play an important role. 
We then use these results to prove Theorem~\ref{t-main0} in Section~\ref{s-theproof}.

\subsection{Preliminaries} \label{s-pre}

The \NP-completeness of {\sc Edge Steiner Tree} on complete graphs follows from the result~\cite{Ka72} that the general problem is \NP-complete: to obtain a reduction add any missing edges and give them sufficiently large weight such that they will never be used in any solution. Bern and Plasman proved 
the following stronger result.

\begin{lemma}[\cite{BP89}]\label{l-complete}
{\sc Edge Steiner Tree} is \NP-complete for complete graphs where every edge has weight~$1$ or~$2$.
\end{lemma}
To subdivide an edge $e=uv$ means to delete $e$ and add a vertex $w$ and edges $uw$ and $vw$. 
Let $r$ be a positive integer.
 To say that $e$ is subdivided $r$ times means that  $e$ is replaced by a path $P_e=uw_1\cdots w_rv$ of $r+1$ edges.
The {\it $r$-subdivision} of a graph $H$ is the graph obtained from $H$ after subdividing each edge exactly $r$ times.
If we say that a graph is a \emph{subdivision} of $H$, then we mean it can be obtained from $H$ using subdivisions (the number of subdivisions can be different for each edge and some edges might not be subdivided at all).
A graph $G$ contains a graph $H$ as a {\it subdivision} if $G$ contains a subdivision of $H$ as a subgraph.

\begin{proposition}
\label{p-hardsupport}
If {\sc Edge Steiner Tree} is \NP-complete for a class~$\mathcal{C}$ of graphs, then for every $r\geq 0$, it is also \NP-complete for the class of $r$-subdivisions of graphs in $\mathcal{C}$.
\end{proposition}

\begin{proof}
Let $(G,w,U,k)$ be an instance of {\sc Edge Steiner Tree} where $G \in \mathcal{C}$. 
Let~$G'$ be the $r$-subdivision of $G$ and for each edge $e$ in $G$, let $P_e$ be the corresponding path on $r+1$ edges in $G'$.
We define an edge weighting $w'$ for $G'$ by letting $w'(e')=w(e)/{(r+1)}$ for each $e\in E(G)$ and for each $e'\in E(P_e)$. 
In any minimum Steiner tree $T'$ for $U$ of $(G',w')$, for each $e \in E$, either all or no edges of $P_e$ are in $T'$.  Then it is easy to see that there is a bijection that preserves weight between minimum Steiner trees of $(G,w)$ for $U$ and minimum Steiner trees of $(G',w')$ for $U$.
\qed\end{proof}

We make the following observation.

\begin{lemma}
\label{l-cbg}
{\sc Edge Steiner Tree} is \NP-complete for complete bipartite graphs.
\end{lemma}

\begin{proof}
As {\sc Edge Steiner Tree} is \NP-complete in general, and the $1$-subdivision of any graph is bipartite, the problem remains 
 \NP-complete for bipartite graphs by Proposition~\ref{p-hardsupport}.
We describe a reduction from the problem on bipartite graphs to the problem on complete bipartite graphs. 
  Let $(G,w,U,k)$ be an instance of {\sc Edge Steiner Tree} where~$G$ is bipartite.  Let $M$ be the sum of the weights of all the edges of $G$ and note that if $k > M$ we can redefine it as $k=M$ without changing the problem.  Let $G'$ be the complete bipartite graph obtained from $G$ by adding all missing edges that preserve the bipartition.  Let $w'$ be an edge weighting for $G'$ where $w'(e)=w(e)$ if $e$ belongs to $G$ and $w'(e)=M+1$ otherwise. 
 Thus $(G',w',U,k)$ is an instance of {\sc Edge Steiner Tree} for complete bipartite graphs and clearly~$G$ has a Steiner tree $T_U$ of weight at most $k$ if and only if $G'$ has a Steiner tree~$T_U'$ of weight at most $k$.
\qed
\end{proof}

The next theorem follows by inspection of the reduction of Garey and Johnson for {\sc Rectilinear Steiner Tree}~\cite{GJ77}. Let $n$ and $m$ be positive integers. An $n \times m$ grid graph has vertex set $\{v^{i,j} \mid  1 \leq i \leq n, 1 \leq j \leq m\}$ and $v^{i,j}$ has neighbours $v^{i-1,j}$ (if $i > 1$), $v^{i+1,j}$ (if $i < n$), $v^{i,j-1}$ (if $j > 1$), and $v^{i,j+1}$ (if $j < m$). 
Think of $v^{1,1}$ as the top-left corner of the grid, and in $v^{i,j}$, $i$ indicates the row of the grid containing the vertex, while $j$ indicates the column.

\begin{theorem}[\cite{GJ77}] \label{t-grid-est}
Unweighted {\sc Edge Steiner Tree}  is \NP-complete for grid graphs.
\end{theorem}

\newcommand{\upv}[1]{\ensuremath{v^{#1}_{\uparrow}}}
\newcommand{\downv}[1]{\ensuremath{v^{#1}_{\downarrow}}}

\begin{figure}[tb]
\begin{minipage}{0.5\textwidth}
\centering
\begin{tikzpicture}[scale=0.85]
\draw (-2,0)--(2,0)--(2,1)--(-3,1)--(-3,2)--(1,2)--(1,1) (-2,0)--(-2,1) (0,0)--(0,1) (-1,1)--(-1,2);
\draw[fill=black]
(-2,0) circle [radius=2pt] (-1,0) circle [radius=2pt] (0,0) circle [radius=2pt] (1,0) circle [radius=2pt] (2,0) circle [radius=2pt] (-3,1) circle [radius=2pt] (-2,1) circle [radius=2pt] (-1,1) circle [radius=2pt] 
(0,1) circle [radius=2pt] (1,1) circle [radius=2pt] (2,1) circle [radius=2pt] (-3,2) circle [radius=2pt] (-2,2) circle [radius=2pt] (-1,2) circle [radius=2pt] (0,2) circle [radius=2pt] (1,2) circle [radius=2pt];
\end{tikzpicture}
\end{minipage}
\begin{minipage}{0.2\textwidth}
\centering
\begin{tikzpicture}[scale=0.85]
\draw(2,0)--(-2,0)--(-2,0.7)--(-2.3,1)--(-3,1)--(-3,2)--(1,2)--(1,1.3)--(1.3,1)--(2,1)--(2,0) (-1,1.7)--(-1.3,2) (-0.7,2)--(-1,1.7)--(-1,1.3)--(-1.3,1)--(-1.7,1)--(-2.3,1) (-1.7,1)--(-2,0.7)
(-0.3,0)--(0,0.3)--(0,0.7)--(-0.3,1)--(-0.7,1)--(-1.3,1) (-0.7,1)--(-1,1.3) (-0.3,1)--(0.3,1)--(0.7,1)--(1,1.3) (0.7,1)--(1.3,1) (0.3,1)--(0,0.7) (0,0.3)--(0.3,0);
\draw[fill=black] (-2,0) circle [radius=2pt] (-1,0) circle [radius=2pt] (-0.3,0) circle [radius=2pt] (0.3,0) circle [radius=2pt] (0,0.3) circle [radius=2pt] (1,0) circle [radius=2pt] (2,0) circle [radius=2pt] 
(-3,1) circle [radius=2pt] (-1.7,1) circle [radius=2pt] (-2.3,1) circle [radius=2pt] (-2,0.7) circle [radius=2pt] (-1.3,1) circle [radius=2pt] (-0.7,1) circle [radius=2pt] (-1,1.3) circle [radius=2pt] (-0.3,1) circle [radius=2pt]
(0.3,1) circle [radius=2pt] (0,0.7) circle [radius=2pt] (0.3,1) circle [radius=2pt] (0.7,1) circle [radius=2pt] (1.3,1) circle [radius=2pt] (1,1.3) circle [radius=2pt] (2,1) circle [radius=2pt] 
(-3,2) circle [radius=2pt] (-2,2) circle [radius=2pt] (-0.7,2) circle [radius=2pt] (-1.3,2) circle [radius=2pt] (-1,1.7) circle [radius=2pt] (0,2) circle [radius=2pt] (1,2) circle [radius=2pt];
\end{tikzpicture}
\end{minipage}
\caption{A wall of height~$2$ and the net-wall obtained by applying a wye-net transformation.}\label{fig:walls1}

\end{figure}
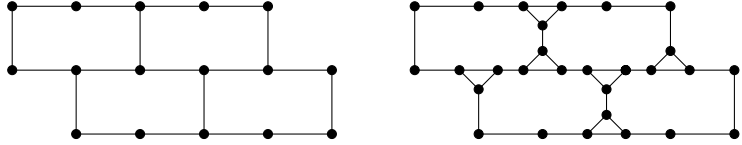

\begin{figure}\hspace{-120pt}
\begin{center}
\begin{tikzpicture}[scale=1]
\draw (-2,0)--(-2,-1)--(4,-1)--(4,0)--(-3,0)--(-3,1)--(4,1)--(4,2)--(-2,2)--(-2,1) (0,-1)--(0,0) (0,1)--(0,2) (2,-1)--(2,0) (2,1)--(2,2)(-1,0)--(-1,1) (1,0)--(1,1) (3,0)--(3,1); \draw[fill=black] (-2,-1) circle [radius=2pt] (-1,-1) circle [radius=2pt] (0,-1) circle [radius=2pt] (1,-1) circle [radius=2pt] (2,-1) circle [radius=2pt] (3,-1) circle [radius=2pt] (4,-1) circle [radius=2pt] (-3,0) circle [radius=2pt] (-2,0) circle [radius=2pt] (-1,0) circle [radius=2pt] (0,0) circle [radius=2pt] (1,0) circle [radius=2pt] (2,0) circle [radius=2pt] (3,0) circle [radius=2pt] (4,0) circle [radius=2pt] (-3,1) circle [radius=2pt] (-2,1) circle [radius=2pt] (-1,1) circle [radius=2pt] (0,1) circle [radius=2pt] (1,1) circle [radius=2pt] (2,1) circle [radius=2pt] (3,1) circle [radius=2pt] (4,1) circle [radius=2pt] (-2,2) circle [radius=2pt] (-1,2) circle [radius=2pt] (0,2) circle [radius=2pt] (1,2) circle [radius=2pt] (2,2) circle [radius=2pt] (3,2) circle [radius=2pt] (4,2) circle [radius=2pt]; \node[above] at (-2,2) {$\downv{1,1}$}; \node[above] at (-1,2) {$\upv{1,2}$}; \node[above] at (0,2) {$\downv{1,2}$}; \node[above] at (1,2) {$\upv{1,3}$}; \node[above] at (2,2) {$\downv{1,3}$}; \node[above] at (3,2) {$\upv{1,4}$}; \node[above] at (4,2) {$\downv{1,4}$}; \node[above left] at (-3,1) {$\downv{2,1}$}; \node[above left] at (-1.9,1) {$\upv{2,1}$}; \node[above] at (-1,1) {$\downv{2,2}$}; \node[above left] at (0.1,1) {$\upv{2,2}$}; \node[above] at (1,1) {$\downv{2,3}$}; \node[above left] at (2.1,1) {$\upv{2,3}$}; \node[above] at (3,1) {$\downv{2,4}$}; \node[above left] at (4.1,1) {$\upv{2,4}$}; \node[above left] at (-3,0) {$\upv{3,1}$}; \node[above] at (-2,0) {$\downv{3,1}$}; \node[above left] at (-0.9,0) {$\upv{3,2}$}; \node[above] at (0,0) {$\downv{3,2}$}; \node[above left] at (1.1,0) {$\upv{3,3}$}; \node[above] at (2,0) {$\downv{3,3}$}; \node[above left] at (3.1,0) {$\upv{3,4}$}; \node[above] at (4,0) {$\downv{3,4}$}; \node[above left] at (-1.9,-1) {$\upv{4,1}$}; \node[above] at (-1,-1) {$\downv{4,2}$}; \node[above left] at (0.1,-1) {$\upv{4,2}$}; \node[above] at (1,-1) {$\downv{4,3}$}; \node[above left] at (2.1,-1) {$\upv{4,3}$}; \node[above] at (3,-1) {$\downv{4,4}$}; \node[above left] at (4.1,-1) {$\upv{4,4}$};
\end{tikzpicture}
\end{center}
\caption{A wall of height~$3$ (obtained from a grid with $n=4$ and $m=4$) and the labelling of its vertices per Lemma~\ref{l-wall}. Note that $v^{n,1}$ and $v^{1,1}$ are exceptional.}\label{fig:walls2}
\end{figure}
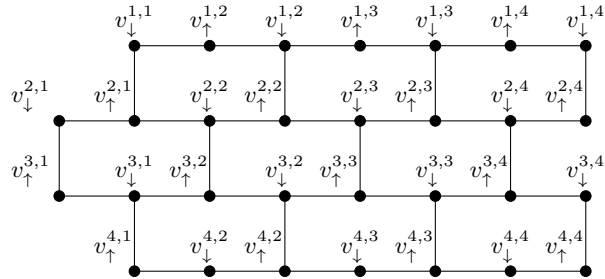

\begin{figure}
\begin{center}
\begin{tikzpicture}[scale=1]
\draw (-4,2)--(4,2) (-4,1)--(5,1) (-4,0)--(5,0) (-4,-1)--(5,-1) (-3,-2)--(5,-2) (-4,2)--(-4,1) (-4,0)--(-4,-1)(-2,2)--(-2,1) (-2,0)--(-2,-1) (0,2)--(0,1) (0,0)--(0,-1)  (2,2)--(2,1) (2,0)--(2,-1) (4,2)--(4,1) (4,0)--(4,-1) (-3,1)--(-3,0) (-3,-1)--(-3,-2) (-1,1)--(-1,0) (-1,-1)--(-1,-2) (1,1)--(1,0) (1,-1)--(1,-2) (3,1)--(3,0) (3,-1)--(3,-2) (5,1)--(5,0) (5,-1)--(5,-2); \draw[fill=black] 
(-3,-2) circle [radius=2pt] (-2,-2) circle [radius=2pt] (-1,-2) circle [radius=2pt] (0,-2) circle [radius=2pt] (1,-2) circle [radius=2pt] (2,-2) circle [radius=2pt] (3,-2) circle [radius=2pt] (4,-2) circle [radius=2pt] (5,-2) circle [radius=2pt] (-4,-1) circle [radius=2pt](-3,-1) circle [radius=2pt] (-2,-1) circle [radius=2pt] (-1,-1) circle [radius=2pt] (0,-1) circle [radius=2pt] (1,-1) circle [radius=2pt] (2,-1) circle [radius=2pt] (3,-1) circle [radius=2pt] (4,-1) circle [radius=2pt] (5,-1) circle [radius=2pt] (-4,0) circle [radius=2pt] (-3,0) circle [radius=2pt] (-2,0) circle [radius=2pt](-1,0) circle [radius=2pt] (0,0) circle [radius=2pt] (1,0) circle [radius=2pt] (2,0) circle [radius=2pt] (3,0) circle [radius=2pt] (4,0) circle [radius=2pt] (5,0) circle [radius=2pt] (-4,1) circle [radius=2pt] (-3,1) circle [radius=2pt] (-2,1) circle [radius=2pt] (-1,1) circle [radius=2pt] (0,1) circle [radius=2pt] (1,1) circle [radius=2pt] (2,1) circle [radius=2pt] (3,1) circle [radius=2pt] (4,1) circle [radius=2pt] (5,1) circle [radius=2pt] (-4,2) circle [radius=2pt] (-3,2) circle [radius=2pt] (-2,2) circle [radius=2pt] (-1,2) circle [radius=2pt] (0,2) circle [radius=2pt] (1,2) circle [radius=2pt] (2,2) circle [radius=2pt] (3,2) circle [radius=2pt] (4,2) circle [radius=2pt]; \node[above] at (-4,2) {$\downv{1,1}$}; \node[above] at (-3,2) {$\upv{1,1}$}; \node[above] at (-2,2) {$\downv{1,2}$}; \node[above] at (-1,2) {$\upv{1,2}$}; \node[above] at (0,2) {$\downv{1,3}$}; \node[above] at (1,2) {$\upv{1,3}$}; \node[above] at (2,2) {$\downv{1,4}$}; \node[above] at (3,2) {$\upv{1,4}$}; \node[above] at (4,2) {$\downv{1,5}$}; \node[above left] at (-3.9,1) {$\upv{2,1}$}; \node[above] at (-3,1) {$\downv{2,1}$}; \node[above left] at (-1.9,1) {$\upv{2,2}$}; \node[above] at (-1,1) {$\downv{2,2}$}; \node[above left] at (0.1,1) {$\upv{2,3}$}; \node[above] at (1,1) {$\downv{2,3}$}; \node[above left] at (2.1,1) {$\upv{2,4}$}; \node[above] at (3,1) {$\downv{2,4}$}; \node[above left] at (4.1,1) {$\upv{2,5}$}; \node[above] at (5,1) {$\downv{2,5}$}; \node[above] at (-4,0) {$\downv{3,1}$}; \node[above left] at (-2.9,0) {$\upv{3,1}$}; \node[above] at (-2,0) {$\downv{3,2}$}; \node[above left] at (-0.9,0) {$\upv{3,2}$}; \node[above] at (0,0) {$\downv{3,3}$}; \node[above left] at (1.1,0) {$\upv{3,3}$}; \node[above] at (2,0) {$\downv{3,4}$}; \node[above left] at (3.1,0) {$\upv{3,4}$}; \node[above] at (4,0) {$\downv{3,5}$}; \node[above left] at (5.1,0) {$\upv{3,5}$}; \node[above left] at (-3.9,-1) {$\upv{4,1}$}; \node[above] at (-3,-1) {$\downv{4,1}$}; \node[above left] at (-1.9,-1) {$\upv{4,2}$}; \node[above] at (-1,-1) {$\downv{4,2}$}; \node[above left] at (0.1,-1) {$\upv{4,3}$}; \node[above] at (1,-1) {$\downv{4,3}$}; \node[above left] at (2.1,-1) {$\upv{4,4}$}; \node[above] at (3,-1) {$\downv{4,4}$}; \node[above left] at (4.1,-1) {$\upv{4,5}$}; \node[above] at (5,-1) {$\downv{4,5}$}; \node[above left] at (-2.9,-2) {$\upv{5,1}$}; \node[above] at (-2,-2) {$\downv{5,2}$}; \node[above left] at (-0.9,-2) {$\upv{5,2}$}; \node[above] at (0,-2) {$\downv{5,3}$};\node[above left] at (1.1,-2) {$\upv{5,3}$}; \node[above] at (2,-2) {$\downv{5,4}$}; \node[above left] at (3.1,-2) {$\upv{5,4}$}; \node[above] at (4,-2) {$\downv{5,5}$}; \node[above left] at (5.1,-2) {$\upv{5,5}$};
\end{tikzpicture}
\end{center}
\caption{A wall of height~$4$ (obtained from a grid with $n=5$ and $m=5$) and the labelling of its vertices per Lemma~\ref{l-wall}. Note that $v^{n,1}$ and $v^{1,m}$ are exceptional.} \label{fig:walls3}
\end{figure}

A \emph{wall} is a graph which can be thought of as a hexagonal grid. See \figurename~\ref{fig:walls1}, \ref{fig:walls2}, and~\ref{fig:walls3} for three examples of walls of different {\it heights} and {\it widths}.
We refer to~\cite{Ch15} for a formal definition.

From a wall of height~$h$ we obtain a {\it net-wall} by doing the following: for each 
wall vertex~$u$ with three neighbours $v_1$, $v_2$, $v_3$,
replace $u$ and its incident edges with three new vertices $u_1$, $u_2$, $u_3$ and edges $u_1v_1$, $u_2v_2$, $u_3v_3$, $u_1u_2$, $u_1u_3$, $u_2u_3$. We call this a \emph{wye-net transformation}, reminiscent of the well-known wye-delta transformation (see~\cite{Tr89}).
Note that a net-wall is $K_{1,3}$-free but contains an induced {\it net}, which is the graph obtained from a triangle on vertices $a_1,a_2,a_3$ and three new vertices $b_1,b_2,b_3$ by adding the edge $a_ib_i$ for $i=1,2,3$.
We have two results related to these classes.

\begin{lemma}\label{l-wall}
For every $r\geq 0$, {\sc Edge Steiner Tree} is \NP-complete for $r$-subdivisions of walls.
\end{lemma}
\begin{proof}
We reduce from unweighted {\sc Edge Steiner Tree} on grid graphs, which is \NP-complete by Theorem~\ref{t-grid-est}.
Let $(G,U,k)$ be an instance of unweighted {\sc Edge Steiner Tree} where $G$ is an $n \times m$ grid graph.
By adding a few rows and columns to the side of the grid, we may assume that the two rows and columns forming the boundary grid are free of terminals and that the outer row and column will not be used by any optimal solution to $(G,U,k)$. We call such an instance \emph{clean}. 
We call a Steiner tree \emph{neat} if it avoids the outer row and column of the grid or the outer row and 
column
of the wall we construct in the reduction. 
(The outer row and column of a wall means those vertices corresponding to the outer row and column of the grid from which it is obtained by the splitting of the vertices as explained below.)
Note that, because the instance is clean, asking for a neat Steiner tree is not a restriction.

From $G$, we obtain a graph $W$ as follows. See \figurename~\ref{fig:walls2} and ~\ref{fig:walls3} for examples.
Two vertices of $G$ are exceptional: $v^{n,1}$ is always exceptional, $v^{1,m}$ is exceptional if $n$ is odd, and $v^{1,1}$ is exceptional if $n$ is even.  For every vertex $v^{i,j}$ of $G$ that is not exceptional,~$W$ contains vertices $\upv{i,j}$ and $\downv{i,j}$ that are joined by an edge.  We call these edges \emph{new}.
We also add to~$W$ vertices $\upv{n,1}$, and $\downv{1,m}$ (if $v^{1,m}$ is exceptional) or $\downv{1,1}$ (if $v^{1,1}$ is exceptional).
We add an edge from $\downv{i,j}$ to $\upv{i+1,j}$, for $1 \leq i \leq n-1$, $1 \leq j \leq m$.
For $1 \leq i \leq n$, $1 \leq j \leq m-1$, if $i$ is odd and $n$ is even or if $i$ is even and $n$ is odd, we add an edge from $\downv{i,j}$ to $\upv{i,j+1}$, and, otherwise, we add an edge from $\upv{i,j}$ to~$\downv{i,j+1}$.  The edges that are not new are \emph{original}.

We note that $W$ is a wall obtained from $G$ by splitting each vertex in two (except the exceptional vertices that lie in a corner of the grid),
 and that there is a bijection between the original edges of $W$ and the edges of $G$.  We define an edge weighting $w'$ for $W$ by letting the weight of each original edge be 1 and the weight of each new edge be a rational number $\varepsilon$, where $\varepsilon > 0$ is chosen so that the sum of the weights of all new edges is less than 1.
We define a set of terminals $U'$ for $W$: if $v^{i,j}$ is in $U$, then $U'$ contains each of $\downv{i,j}$ and $\upv{i,j}$ (note that both vertices exist because the instance is clean). Observe that $W$ is also clean.

We claim that there is a neat Steiner tree of $k$ edges in $G$ for terminal set $U$ if and only if there is a neat Steiner tree of weight $k+(k+1)\varepsilon$ in $(W,w')$ for terminal set $U'$.  
Indeed, any neat Steiner tree $T$ in $G$ for terminal set $U$ of $k$ edges corresponds naturally to a neat Steiner tree $T'$ for $U'$ in $(W,w')$ of weight 
$k+(k+1)\varepsilon$
by adding for each vertex of $T$ the corresponding new edge to $T'$ (note that the neatness of $T$ implies that the corresponding new edges exist). 
Conversely, any neat Steiner tree $T'$ for $U'$ in $(W,w')$ of weight  
$k+(k+1)\varepsilon$
corresponds naturally to a neat Steiner tree $T$ for $U$ in $G$ of $k$ edges by removing all new edges from $T'$.
Effectively, this mimics the splitting and contraction operations which can be seen as the way in which we obtain $W$ from $G$ and vice versa.

The lemma now follows immediately from Proposition~\ref{p-hardsupport}.\qed
\end{proof}

The next lemma has a similar proof.

\begin{lemma}
\label{l-netwall}
For every $r\geq 0$, {\sc Edge Steiner Tree} is \NP-complete for $r$-subdivisions of net-walls.
\end{lemma}

\begin{proof}
We reduce from 
 {\sc Edge Steiner Tree} on walls, which is \NP-complete by Lemma~\ref{l-wall}. Consider an instance $(W,w,U,k)$ of {\sc Edge Steiner Tree} where~$W$ is a wall. 
 By the construction of Lemma~\ref{l-wall}, we may assume that the two outer rows and 
 columns of the wall are free of terminals and the outer row and 
 column will not be used by any optimal solution to $(W,w,U,k)$. We call such an instance \emph{clean}. We call a Steiner tree \emph{neat} if it avoids the outer row and 
 column. Note that, because the instance is clean, asking for a neat Steiner tree is not a restriction.
 We now apply a wye-net transformation to $W$ to obtain a {\it net-wall}~$N$.

The edges of the triangles added through the wye-net transformations are called \emph{new} and all other edges of $N$ are \emph{original} (since they admit a bijection to the edges of $W$).
We create an edge weighting $w'$ on $N$. For each original edge $e$, $w'(e)=w(e)$.  
Let $s$ be the smallest weight of an edge of $(W,w)$.
Let the weight of each new edge be a rational number $\varepsilon$, where $\varepsilon > 0$ is chosen so that the sum of the weights of all new edges is less than $s$.
We define a set of terminals $U'$ for $N$.
We note that each vertex $v$ of $W$ that is in $U$ corresponds to 
a set of three vertices of $N$ that form a triangle,
by the fact that the instance is clean.
If $v$ is in $U$, then
we add to $U'$ 
all three vertices of the triangle.
Note that $N$ is also clean.

We claim that there is a neat Steiner tree of weight $k$ in $(W,w)$ for terminal set $U$  if and only if there is a neat Steiner tree of weight 
$k+2(k+1)\varepsilon$
in $(N,w')$ for terminal set $U'$.
Indeed, any neat Steiner tree $T$ in $W$ for terminal set $U$ of weight at most $k$ corresponds naturally to a neat Steiner tree $T'$ for $U'$ in $(N,w')$ of weight 
$k+2(k+1)\varepsilon$
adding for each vertex of the $k+1$ vertices of $T$ any two edges of the corresponding triangle to $T'$ (note that the neatness of $T$ implies that the corresponding triangles exist).
Conversely, any neat Steiner tree $T'$ for $U'$ in $(N,w')$ of weight 
$k+2(k+1)\varepsilon$
corresponds naturally to a neat Steiner tree $T$ for $U$ in $(W,w)$ of weight at most $k$ by removing all new edges from $T'$.
 
The lemma now follows immediately from Proposition~\ref{p-hardsupport}.
\qed\end{proof}

\subsection{Treewidth and Implications}\label{s-treewidth}

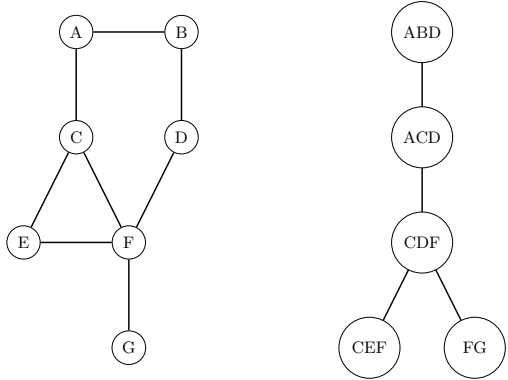
\begin{figure}[tb]
\begin{center} 
\begin{tabular}{ccc}
\begin{minipage}{0.25\textwidth}
\centering
\scalebox{0.7}{
{\begin{tikzpicture}[scale=1,rotate=0]
\GraphInit[vstyle=Normal]
\Vertex[x=1,y=6]{A}
\Vertex[x=3,y=6]{B}
\Vertex[x=1,y=4]{C}
\Vertex[x=3,y=4]{D}
\Vertex[x=0,y=2]{E}
\Vertex[x=2,y=2]{F}
\Vertex[x=2,y=0]{G}

\Edges(A,B,D,F,E,C,A)
\Edges(C,F)
\Edges(F,G)
\end{tikzpicture}}}
\end{minipage}
&
\hspace{1cm}
&
\begin{minipage}{0.25\textwidth}
\centering
\scalebox{0.7}{
{\begin{tikzpicture}[scale=1,rotate=0]
\GraphInit[vstyle=Normal]
\SetVertexNormal[MinSize=33pt]
\Vertex[x=1,y=6]{ABD}
\Vertex[x=1,y=4]{ACD}
\Vertex[x=1,y=2]{CDF}
\Vertex[x=0,y=0]{CEF}
\Vertex[x=2,y=0]{FG}

\Edges(ABD,ACD,CDF,CEF)
\Edges(CDF,FG)
\end{tikzpicture}}}
\end{minipage}\\
\end{tabular}
\end{center}
\caption{A graph, and a tree decomposition of the graph with width~$2$.}
\label{fig:tw}
\end{figure}

A \emph{tree decomposition} of a graph~$G=(V,E)$ is a tree~$T$ whose vertices, which are called {\it nodes}, are subsets of~$V$ and has the following properties: for each~$v\in V$, the nodes of $T$ that contain~$v$ induce a connected subgraph with at least one node, and, for each edge~$vw\in E$, there is at least one node of $T$ that contains~$v$ and~$w$. We refer to \figurename~\ref{fig:tw} for an example.

The sets of vertices of $G$ that form the nodes of $T$ are called \emph{bags}.
The {\it width} of $T$ is one less than the size of its largest bag. 
The \emph{treewidth} of~$G$ is the minimum width of its tree decompositions.
A graph class ${\cal G}$ has {\it bounded treewidth} if there exists a constant~$c$ such that each graph in ${\cal G}$ has treewidth at most~$c$;
otherwise ${\cal G}$ has {\it unbounded treewidth}.
As trees with at least one edge
form exactly the class of connected graphs with treewidth~$1$,
the treewidth of a graph can be seen as a measure that indicates how close a graph is to being a tree. 

We recall the following well-known result.

\begin{lemma}[\cite{BCKN15,CMZ12}]\label{l-treewidth}
{\sc Edge Steiner Tree} can be solved in polynomial time on every graph class of bounded treewidth.
\end{lemma}

For some of our poofs we need the well-known Robertson-Seymour Grid-Minor Theorem (also called the Excluded Grid Theorem), which can be formulated for walls.

\begin{theorem}[\cite{RS86}]\label{l-grid}
For every integer~$h$, there exists a constant $c_h$ such that a graph has treewidth at least $c_h$ if and only if it
contains a wall of height $h$ as a subdivision.
\end{theorem}
 
We will use two lemmas, both of which follow immediately from Theorem~\ref{l-grid}. 

\begin{lemma}\label{l-sw}
For every $r\geq 0$, the class of $r$-subdivisions of walls has unbounded treewidth.
\end{lemma} 

\begin{lemma}\label{l-sw2}
For every $r\geq 0$, the class of $r$-subdivisions of net-walls has unbounded treewidth.
\end{lemma} 

\subsection{The Proof}\label{s-theproof}

We are now ready to prove Theorem~\ref{t-main0}.

\medskip
\noindent
{\bf Theorem~\ref{t-main0} (restated).}
{\it Let $H_1$ and $H_2$ be two graphs. If one of the following cases holds:
\begin{enumerate}
\item $H_1=K_r$ or $H_2=K_r$ for some $r\in \{1,2\}$;
\item $H_1=K_3$ and $H_2=K_{1,3}$, or vice versa;
\item $H_1=K_r$ for some $r\geq 3$ and $H_2=P_3$, or vice versa; or
\item $H_1=K_r$ for some $r\geq 3$ and $H_2=sP_1$ for some $s\geq 1$, or vice versa.\\[-15pt]
\end{enumerate}
then the treewidth of the class of $(H_1,H_2)$-free graphs is bounded and 
{\sc Edge Steiner Tree} is polynomial-time solvable for $(H_1,H_2)$-free graphs. In all other cases,
the treewidth of $(H_1,H_2)$-free graphs is unbounded and {\sc Edge Steiner Tree} is \NP-complete for $(H_1,H_2)$-free graphs.}

\begin{proof}
First suppose that one of the Cases~1--4 holds.
Let $G$ be an $(H_1,H_2)$-free graph.
First suppose that $H_1=K_r$ for some $r\in \{1,2\}$. Then $G$ has no edges and so has treewidth~$0$.
If $H_1=K_3$ and $H_2=K_{1,3}$, then $G$ has maximum degree at most~$2$, that is, $G$ is the disjoint union of paths and cycles. Hence~$G$ has treewidth at most $2$. If $H_1=K_r$ for some $r\geq 3$, and $H_2=P_3$, then $G$ is the disjoint union of complete graphs, each of size at most~$r-1$. 
Hence $G$ has treewidth at most $r-2$.
Finally, suppose that $H_1=K_r$, for some $r\geq 3$, and $H_2=sP_1$, for some $s\geq 1$. 
Ramsey's Theorem tells us that for every $r\geq 1$ and $s\geq 1$, there exists a constant $R(r,s)$ such that every graph on at least $R(r,s)$ vertices contains a clique on $r$ vertices or an independent set on $s$ vertices. As $G$ is $(K_r,sP_1)$-free, we find that the number of vertices of $G$ is bounded by $R(r,s)-1$. Hence, $G$ has treewidth at most $R(r,s)-2$.
In order to complete the proof of the first statement of the theorem we can now apply Lemma~\ref{l-treewidth} and find that {\sc Edge Steiner Tree} is polynomial-time solvable for $(H_1,H_2)$-free graphs.\footnote{Note that the graph under consideration in Cases~1--4 has a very restricted structure: any large connected component (if it exists) is either a path or a cycle. Hence, we can also solve the problem directly instead of applying Lemma~\ref{l-treewidth}.} 

\medskip
\noindent
Now suppose that none of Cases~1--4 apply. We will prove that the treewidth of $(H_1,H_2)$-free graphs is unbounded and that {\sc Edge Steiner Tree} is \NP-complete for the class of $(H_1,H_2)$-free graphs.

First suppose that neither $H_1$ nor $H_2$ is a complete graph. Then the class of $(H_1,H_2)$-free graphs contains the class of all complete graphs. As the treewidth of a complete graph $K_r$ is readily seen to be equal to $r-1$, the class of complete graphs, and thus the class of $(H_1,H_2)$-free graphs, has unbounded treewidth.
Moreover, by Lemma~\ref{l-complete}, {\sc Edge Steiner Tree} is \NP-complete for the class of $(H_1,H_2)$-free graphs.
From now on, assume that $H_1=K_r$ for some $r\geq 1$. As Case~1 does not apply, we find that $r\geq 3$.

Suppose that $H_2$ contains a cycle $C_s$ as an induced subgraph for some $s\geq 3$.
As $H_1=K_r$ for some $r\geq 3$, the class of $(H_1,H_2)$-free graphs contains the class of $(C_3,C_s)$-free graphs. As the latter graph class contains the class of $(s+1)$-subdivisions of walls, which have unbounded treewidth due to Lemma~\ref{l-sw}, the class of $(H_1,H_2)$-free graphs has unbounded treewidth.
Moreover, by Lemma~\ref{l-wall}, {\sc Edge Steiner Tree} is \NP-complete for the class of $(H_1,H_2)$-free graphs.

Note that if $H_2$ contains a cycle as a subgraph, then it also contains a cycle as an induced subgraph.
So now suppose that $H_2$ contains no cycle, that is, $H_2$ is a forest.
First assume that $H_2$ contains an induced 
$P_1+P_2$. Recall that $H_1=K_r$ for some $r\geq 3$. 
Then the class of  $(H_1,H_2)$-free graphs contains the class of complete bipartite graphs. 
As this class has unbounded treewidth, the class of $(H_1,H_2)$-free graphs has unbounded treewidth.
Moreover, by Lemma~\ref{l-cbg}, {\sc Edge Steiner Tree} is \NP-complete for the class of $(H_1,H_2)$-free graphs.
From here on we assume that $H_2$ is a $(P_1+P_2)$-free forest.

Suppose that $H_2$ has a vertex of degree at least~$3$. In other words, as $H_2$ is a forest, the claw $K_{1,3}$ is an induced subgraph of $H_2$.  Recall that $H_1=K_r$ for some $r\geq 3$. 
First assume that $r=3$. As Case~2 does not apply, $H_2$ properly contains an induced $K_{1,3}$.  
As $H_2$ is a $(P_1+P_2)$-free forest, this means that $H_2=K_{1,s}$ for some $s\geq 4$. Then the class of $(H_1,H_2)$-free graphs contains the class of walls. As the latter class has unbounded treewidth due to Lemma~\ref{l-sw}, the class of $(H_1,H_2)$-free graphs has unbounded treewidth. Moreover, by Lemma~\ref{l-wall}, {\sc Edge Steiner Tree} is \NP-complete for the class of $(H_1,H_2)$-free graphs.
Now assume that $r\geq 4$. Then the class of $(H_1,H_2)$-free graphs contains the class of net-walls. As the latter graph class has unbounded treewidth due to Lemma~\ref{l-sw2}, the class of $(H_1,H_2)$-free graphs has unbounded treewidth. Moreover, by Lemma~\ref{l-netwall}, {\sc Edge Steiner Tree} is \NP-complete for the class of $(H_1,H_2)$-free graphs.

From the above we may assume that $H_2$ 
does not contain any vertex of degree~$3$. This means that $H_2$ is a linear forest, that is, a disjoint union of paths.
As Case~4 does not apply, $H_2$ has an edge. 
Every $(P_1+P_2)$-free linear forest with an edge is either a $P_2$ or a $P_3$.  
However,
this is not possible, as Case~1 (with the roles of $H_1$ and $H_2$ reversed) and Case~3 do not apply.
We conclude that this case cannot happen.
\qed
\end{proof}

\section{Proof of Theorem~\ref{t-main}}\label{s-theproof1}

In this section we give a proof of our second dichotomy. We state useful past results in Section~\ref{s-lit} followed by some new results for line graphs and $P_4$-free graphs in Section~\ref{s-new}.
Then, in Section~\ref{s-final}, we show how to combine these results to obtain the proof of Theorem~\ref{t-main}.

\subsection{Known Results}\label{s-lit}

The first result we need is due to Brandst\"{a}dt and M\"uller. A graph is {\it chordal bipartite} if it has no induced cycles of length 3 or of length at least 5; that is, a graph is chordal bipartite if it is $(C_3,C_5,C_6,\ldots)$-free.

\begin{theorem}[\cite{BM87}]\label{t-cb}
The unweighted {\sc Vertex Steiner Tree} problem is \NP-complete for chordal bipartite graphs.
\end{theorem}

The second result that we need is due to Farber, Pulleyblank and White. A graph is {\it split} if its vertex set can be partitioned into a clique and an independent set. It is well known that the class of split graphs coincides with the class of $(2P_2,C_4,C_5)$-free graphs~\cite{FH77}.

\begin{theorem}[\cite{FPW85}]\label{t-split}
The unweighted {\sc Vertex Steiner Tree} problem is \NP-complete for split graphs.
\end{theorem}

\subsection{New Results}\label{s-new}

We start with the following lemma.

\begin{lemma}
\label{l-line}
The unweighted {\sc Vertex Steiner Tree} problem is \NP-complete for line graphs.
\end{lemma}

\begin{proof}
By Theorem~\ref{t-grid-est}, unweighted {\sc Edge Steiner Tree} is \NP-complete. 
Let $(G,U,k)$ be an instance of this problem.  From $G$ we construct a new graph $G'$ by introducing a new vertex $v_u$ for each terminal $u\in U$, which we make only adjacent to~$u$. We let $U'$ consist of all these new vertices.
We observe that $G'$ has a Steiner tree $T'$ for $U'$ with at most $k+|U|$ edges if and only if $G$ has a Steiner tree $T$ for $U$ with at most  $k$ edges. 

We now consider the line graph $L(G')$ of $G'$ with set of terminals $U^*=\{uv_u\ |\; u\in U\}$; this is a set of edges in $G'$ and a set of vertices in $L(G')$.  To complete the proof, we show that  $G'$ has a Steiner tree for $U'$ on, say, $\ell$ edges if and only if $L(G')$ has a Steiner tree for $U^*$ on $\ell$ vertices. We first note that the edge set $E'$ of a Steiner tree for $U'$ of $G'$ must contain the set $U^*$. Further, $E'$, considered as a set of vertices of $L(G')$, induces a connected subgraph and has $|E'|=\ell$ vertices. Conversely, if there is a Steiner tree for $U^*$ in $L(G')$ on $\ell$ vertices, then these vertices, considered as edges in $G'$, form a Steiner tree for $U'$ in $G'$.\qed
\end{proof}

Recall that a subgraph $G'$ of a graph $G$ is spanning if $V(G')=V(G)$.
Let $G_1$ and $G_2$ be two graphs.
The {\it join} operation adds an edge between every vertex of $G_1$ and every vertex of $G_2$. The {\it disjoint union} operation
takes the disjoint union of $G_1$ and $G_2$.
A graph $G$ is a {\it cograph} if $G$ can be generated from $K_1$ by a sequence of join and disjoint union operations.
A graph is a cograph if and only if it is $P_4$-free~\cite{CLS81}.
This implies the following well-known lemma.  

\begin{lemma}\label{l-p4}
Every connected $P_4$-free graph on at least two vertices has a spanning complete bipartite subgraph.
\end{lemma}

Let $G$ be a graph.
For a set $S$, the graph $G[S]=(S,\{uv\in E(G)\; u,v\in S\})$ denotes the subgraph of $G$ {\it induced by} $S$. Note that $G[S]$ can be obtained from $G$ by deleting every vertex of $V(G)\setminus S$. If $G$ has a vertex weighting $w$, then $w(S)=\sum_{u\in S}w(u)$ denotes the {\it weight} of $S$.

\begin{lemma}\label{l-sp1p4}
For every $s\geq 0$, {\sc Vertex Steiner Tree} can be solved in time
$O(n^{2s^2-s+5})$ for connected $(sP_1+P_4)$-free graphs on $n$ vertices.
\end{lemma}

\begin{proof}
Let $s \geq 0$ be an integer.
Let $G=(V,E)$ be a connected $(sP_1+P_4)$-free graph with a vertex weighting $w: V\rightarrow \mathbb{Q}^+$ and set of terminals $U$.
We show how to solve the optimization version of {\sc Vertex Steiner Tree} on $G$.
Let $R\subseteq V\setminus U$ be such that 
$G[U\cup R]$ is connected and, subject to this condition, $U\cup R$ has minimum weight $w(U\cup R)$. 
Thus any spanning tree of $G[U\cup R]$ is an optimal solution.  Let us consider the possible size of $R$.

First suppose that $G[U\cup R]$ is $P_4$-free.  Then, by Lemma~\ref{l-p4}, $G[U\cup R]$ has a spanning complete bipartite subgraph.  That is, there is a bipartition $(A,B)$ of $U\cup R$ such that every vertex in $A$ is joined to every vertex in $B$.
We may assume without loss of generality that $|U|\geq 2$. Then $|U\cup R|\geq 2$, and thus neither $A$ nor $B$ is the empty set. If $U$ intersects both $A$ and $B$, then $G[U]$ is connected and $|R|=0$.  So let us assume that $U \subseteq A$, and so $R \supseteq B$.  Then $R \cap A = \emptyset$ since $G[U \cup B]$ is connected.  As we know that every vertex in $A=U$ is joined to every vertex in $B=R$, we find that $|R|=1$.

Suppose instead that $G[U\cup R]$ contains an induced path $P$ on four vertices.  We call the connected components of $G[U]$ \emph{bad} if they do not intersect $P$ or the neighbours of $P$ in $G$.  There are at most $s-1$ bad components; else, $G$ contains an $sP_1+P_4$.  Let $U^\ast$ be a subset of $U$ that includes one vertex from each of these bad components.  Then each vertex of $G[U \cup R]$ belongs either to $U$ or $P$ or is an internal vertex of a shortest path in $G[U \cup R]$ from $P$ to a vertex of $U^\ast$.  The number of internal vertices in such a shortest path is at most $2s+1$; else, the path contains an induced $sP_1+P_4$.  As $R$ is a subset of the union of $V(P)$ and the sets of these internal vertices, we find that $|R| \leq 4 + (2s+1)(s-1) = 2s^2-s+3$.

So in all cases $R$ contains at most $2s^2-s+3$ vertices and our algorithm is just to consider every such set $R$ and check, in each case, whether $G[U\cup R]$ is connected.  Our solution is the smallest set found that satisfies the connectivity constraint.
As there are $O(n^{2s^2-s+3})$ sets to consider, and  checking  connectivity takes $O(n^2)$ time, the algorithm requires  $O(n^{2s^2-s+5})$ time.
\qed
\end{proof}

\subsection{The Proof}\label{s-final}

We are now ready to prove our second dichotomy.

\medskip
\noindent
{\bf Theorem~\ref{t-main} (restated)}
{\it Let $H$ be a graph. If $H$ is an induced subgraph of $sP_1+P_4$ for some $s\geq 0$, then {\sc Vertex Steiner Tree} is polynomial-time solvable for $H$-free graphs, otherwise even unweighted {\sc Vertex Steiner Tree} is \NP-complete.}

\begin{proof}
If $H$ has a cycle, then we apply Theorem~\ref{t-cb} or Theorem~\ref{t-split}.
Hence, we may assume that $H$ has no cycle, so $H$ is a forest. If $H$ contains a vertex of degree at least~$3$, then the class of $H$-free graphs contains the class of claw-free graphs, which in turn contains the class of line graphs. Hence, we can apply Lemma~\ref{l-line}. Thus we may assume that $H$ is a linear forest. If $H$ contains a connected component with at least five vertices or two connected components with at least two vertices each, then the class of $H$-free graphs contains the class of $2P_2$-free graphs. Hence, we can apply Theorem~\ref{t-split}.
It remains to consider the case where $H$ is an induced subgraph of $sP_1+P_4$ for some $s\geq 0$, for which we can apply Lemma~\ref{l-sp1p4}. \qed
\end{proof}

\section{Conclusions}\label{s-con}

We presented complexity dichotomies both for {\sc Edge Steiner Tree} restricted to $(H_1,H_2)$-free graphs and for {\sc Vertex Steiner Tree}
for $H$-free graphs. The latter dichotomy also holds for the unweighted variant, in which case the problems {\sc Edge Steiner Tree}
 and {\sc Vertex Steiner Tree} are equivalent due to Observation~\ref{o-equivalent}.
 
 \begin{figure}
\begin{minipage}[c]{0.5\textwidth}
\begin{tikzpicture}[scale=1]
\draw (0,1)--(-1,1)--(-2,0)--(-1,-1)--(2,-1) (-2,0)--(1,0); \draw[fill=black] (-1,1) circle [radius=2pt] (0,1) circle [radius=2pt] (-2,0) circle [radius=2pt] (-1,0) circle [radius=2pt] (0,0) circle [radius=2pt] (1,0) circle [radius=2pt] (-1,-1) circle [radius=2pt] (0,-1) circle [radius=2pt] (1,-1) circle [radius=2pt] (2,-1) circle [radius=2pt];
\end{tikzpicture}
\end{minipage}
\qquad
\begin{minipage}[c]{0.5\textwidth}
\begin{tikzpicture}[scale=1] 
\draw (0,1)--(-1,1)--(-1,-1)--(2,-1) (-1,0)--(1,0)(-1,1)to[out=240,in=120](-1,-1); \draw[fill=black] (-1,1) circle [radius=2pt] (0,1) circle [radius=2pt] (-1,0) circle [radius=2pt] (0,0) circle [radius=2pt] (1,0) circle [radius=2pt] (-1,-1) circle [radius=2pt] (0,-1) circle [radius=2pt] (1,-1) circle [radius=2pt] (2,-1) circle [radius=2pt];
\end{tikzpicture}
\end{minipage}
\caption{The graph $S_{2,3,4}$ (left) and the line graph of $S_{2,3,4}$ (right).}\label{f-st}
\end{figure}
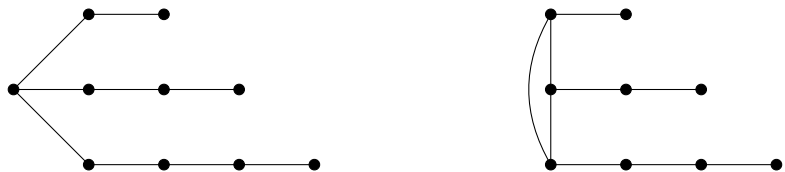
 
\noindent
In particular, we observed that {\sc Edge Steiner Tree} can be solved in polynomial time for $(H_1,H_2)$-free graphs if and only if the class of $(H_1,H_2)$-free graphs has bounded treewidth (assuming \sP$\neq$ \NP). It is a natural to ask whether {\sc Edge Steiner Tree} can be solved in polynomial time on a hereditary graph class ${\cal G}$ characterized by a {\it finite} set ${\cal F}_{\cal G}$ of forbidden induced subgraphs if and only if ${\cal G}$ has bounded treewidth (assuming $\sP\neq \NP$). If ${\cal F}_{\cal G}$ does not contain any complete graph or induced subgraph of a complete bipartite graph, then ${\cal G}$ contains the class of complete graphs or complete bipartite graphs, respectively. Then ${\cal G}$ has unbounded treewidth, and moreover {\sc Edge Steiner Tree} is \NP-complete by Lemma~\ref{l-complete} or~\ref{l-cbg}, respectively. The graph~$S_{h,i,j}$, for $1\leq h\leq i\leq j$, is the \emph{subdivided claw}, which is the tree with one vertex~$x$ of degree~$3$ and exactly three leaves, which are at distance~$h$,~$i$ and~$j$ from~$x$, respectively; see \figurename~\ref{f-st} for an example. Let ${\cal S}$ be the class of graphs, every connected component of which is either a subdivided claw or path. Suppose that ${\cal F}_{\cal G}$ does not contain any graph from ${\cal S}$ as an induced subgraph. As ${\cal F}_{\cal G}$ is finite, we can take a sufficiently large value of $r$ such that ${\cal G}$ contains the class of $r$-subdivisions of walls. Then ${\cal G}$ has unbounded treewidth by Lemma~\ref{l-sw} and {\sc Edge Steiner Tree} is \NP-complete for ${\cal G}$ by Lemma~\ref{l-wall}. Let ${\cal T}$ be the class of line graphs of graphs of ${\cal S}$. By repeating the previous arguments with Lemmas~\ref{l-sw2} and~\ref{l-netwall} instead of Lemmas~\ref{l-sw} and~\ref{l-wall}, respectively, we find that ${\cal G}$ has unbounded treewidth and that {\sc Edge Steiner Tree} is \NP-complete for ${\cal G}$ if ${\cal F}_{\cal G}$ does not contain a graph from ${\cal T}$.

To summarize the above, a class of graphs ${\cal G}$ with finite ${\cal F}_{\cal G}$ has unbounded treewidth and {\sc Edge Steiner Tree} is \NP-complete for ${\cal G}$ if ${\cal F}_{\cal G}$ does not contain any complete graph, or any induced subgraph of a compete bipartite graph, or any graph from ${\cal S}$, or any graph from ${\cal T}$. In a very recent arXiv paper~\cite{LR}, Lozin and Razgon showed that a class ${\cal G}$ with finite ${\cal F}_{\cal G}$ has bounded treewidth if ${\cal F}_{\cal G}$ contains a complete graph, an induced subgraph of a complete bipartite graph, a graph from ${\cal S}$ and a graph from ${\cal T}$. Recall that {\sc Edge Steiner Tree} is polynomial-time solvable for graphs of bounded treewidth (Lemma~\ref{l-treewidth}). Hence, the result of Lozin and Ragzon implies that {\sc Edge Steiner Tree} is polynomial-time solvable on a hereditary graph class ${\cal G}$ with finite ${\cal F}_{\cal G}$  {\it if and only if} ${\cal G}$ has bounded treewidth (assuming P $\neq$ \NP).

The following result shows that the situation changes if ${\cal F}_{\cal G}$ is infinite.

\begin{theorem}\label{t-un}
There exists a hereditary graph class ${\cal G}$ of unbounded treewidth for which {\sc Edge Steiner Tree} can be solved in polynomial time.
\end{theorem}

\begin{proof}
Let ${\cal G}$ consist of graphs $G$ of maximum degree at most~$3$ such that every path between any two degree-$3$ vertices in $G$ has at least $2^{r}$ vertices, where $r$ is the number of degree-$3$ vertices in $G$. As deleting a vertex neither 
increases the maximum degree of a graph
nor shortens any path between a pair of degree-$3$ vertices,
${\cal G}$ is hereditary (note that ${\cal G}$ is also closed under taking edge deletion).
As ${\cal G}$ contains subdivisions of walls of arbitrarily large height, the treewidth of ${\cal G}$ is unbounded due to Theorem~\ref{l-grid}.

We solve {\sc Edge Steiner Tree} on an instance $(G,w,U,k)$ with $G\in {\cal G}$~as follows. 
If $G$ has at most one vertex of degree~$3$, then $G$ has treewidth at most~$2$, so we can apply Lemma~\ref{l-treewidth}.
Suppose that $G$ has at least two vertices of degree~$3$. Then we apply the following rules, while possible.

\medskip
\noindent
{\bf Rule 1.} There is a non-terminal $x$ of degree~$2$. Let $xy$ and $xz$ be its two incident edges.
We contract $xy$ and give the new edge weight $w(xy)+w(xz)$. If there was already an edge between $y$ and $z$, then we remove one with largest weight.

\medskip
\noindent
{\bf Rule 2.} There is a terminal  $x$ of degree~$2$ and its neighbours $y$ and $z$ are also terminals. 
Assume $w(xy)\leq w(xz)$. We observe that there is an optimal solution that includes the edge $xy$. Hence, we may contract $xy$ and decrease $k$ by $w(xy)$.

\medskip
\noindent
{\bf Rule 3.} There is a vertex $x$ of degree~$1$. Let $y$ be its neighbour. If $x$ is not a terminal, then remove $x$. Otherwise, contract $xy$ and decrease $k$ by $w(xy)$.

\medskip
\noindent
Let $(G',w',U',k')$ be the resulting instance, which is readily seen to be equivalent to $(G,w,U,k)$.
Then $G'$ has $r$ vertices of degree~$3$ and each vertex of degree at most~$2$ has a neighbour of degree~$3$; otherwise, one of Rules 1--3 applies. So, $G'$ has at most
$4r$ vertices and thus $O(r)$ edges. 
It remains to solve {\sc Edge Steiner Tree} on $(G',w',U',k)$. We do this in $r\cdot 2^{O(r)}$ time by guessing for each edge in~$G'$ if it is in the solution and then verifying the resulting candidate solution. Recall that in the problem definition we assume that $G$ is connected. As $r\geq 2$ and $G$ contains at least two vertices of degree~$3$, this means that $|V(G)|\geq 2^r$. So, the running time is polynomial in $|V(G)|$. 
\qed
\end{proof}

To obtain a dichotomy for {\sc Vertex Steiner Tree} and unweighted {\sc Vertex Steiner Tree} for $(H_1,H_2)$-free graphs, we need to answer several open problems, including the following two. 

\begin{open}\label{o-o3}
Does there exist a pair $(H_1,H_2)$ such that {\sc Vertex Steiner Tree} and unweighted {\sc Vertex Steiner Tree} have different complexities for $(H_1,H_2)$-free graphs?
\end{open}

\noindent
For Open Problem~\ref{o-o3}, it may be prudent to focus on pairs $(H_1,H_2)$ for which the mim-width of $(H_1,H_2)$-free graphs is unbounded. This is due to the aforementioned result of Bergougnoux and Kant{\'{e}}~\cite{BK19}, who proved that {\sc Vertex Steiner Tree} is polynomial-time solvable for graph classes of bounded mim-width for which we can compute a branch decomposition of constant mim-width in polynomial time.

\begin{open}\label{o-o4}
For every integer~$t$, determine the complexity of {\sc Vertex Steiner Tree} for $(K_{1,3},P_t)$-free graphs.
\end{open}

\noindent
For Open Problem~\ref{o-o4} we note that {\sc Vertex Steiner Tree} is polynomial-time solvable for $P_4$-free graphs by Theorem~\ref{t-main}. It is known that 
$(K_{1,3},P_5)$-free graphs have unbounded mim-width~\cite{BHMPP}. Hence, the first open case is where $t=5$.
To obtain an answer to Open Problem~\ref{o-o4}, we need new insights into the structure of $(K_{1,3},P_t)$-free graphs. These insights may also be useful in obtaining new results for other problems, 
such as the {\sc Graph Colouring} problem restricted to $(K_{1,3},P_t)$-free graphs (see~\cite{GJPS17,MPS19}).

\medskip
\noindent
{\it Acknowledgments.} We thanks an anonymous reviewer for a number of helpful suggestions.


\begin{thebibliography}{10}
\bibitem{BK19}
B.~Bergougnoux and M.~M. Kant{\'{e}}.
\newblock More applications of the d-neighbor equivalence: Connectivity and
  acyclicity constraints.
\newblock {\em Proc. ESA 2019, LIPIcs}, 144:17:1--17:14, 2019.

\bibitem{BP89}
M.~W. Bern and P.~E. Plassmann.
\newblock The {S}teiner problem with edge lengths 1 and 2.
\newblock {\em Information Processing Letters}, 32:171--176, 1989.

\bibitem{BBJPP20}
H. L. Bodlaender, N. Brettell, M. Johnson, G. Paesani, and D. Paulusma.
\newblock Steiner trees for hereditary graph classes.
\newblock {\em Proc. LATIN 2020, LNCS}, 12118:613--624, 2020.

\bibitem{BCKN15}
H.~L. Bodlaender, M.~Cygan, S.~Kratsch, and J.~Nederlof.
\newblock Deterministic single exponential time algorithms for connectivity
  problems parameterized by treewidth.
\newblock {\em Information and Computation}, 243:86--111, 2015.

\bibitem{BM87}
A.~Brandst{\"{a}}dt and H.~M{\"{u}}ller.
\newblock The {N}{P}-completeness of {S}teiner {T}ree and {D}ominating {S}et
  for chordal bipartite graphs.
\newblock {\em Theoretical Computer Science}, 53:257--265, 1987.

\bibitem{BCM15}
J. Brault-Baron, F. Capelli, and S. Mengel. Understanding model counting for beta-acyclic CNF-formulas. 
{\em Proc. STACS 2015, LIPIcs}, 30:143–156, 2015.

\bibitem{BHMPP}
N. Brettell, J. Horsfield, A. Munaro, G. Paesani, and D. Paulusma.
\newblock Bounding the mim-width of hereditary graph classes.
\newblock {\em Proc. IPEC 2020, LIPIcs}, to appear.

\bibitem{BTV11}
B.-M. Bui-Xuan, J.A. Telle, and M. Vatshelle.
\newblock Boolean-width of graphs.
\newblock {\em Theoretical Computer Science}, 412:5187--5204, 2011.

\bibitem{CMZ12}
M.~Chimani, P.~Mutzel, and B.~Zey.
\newblock Improved {S}teiner tree algorithms for bounded treewidth.
\newblock {\em Journal of Discrete Algorithms}, 16:67--78, 2012.

\bibitem{CLS81}
D.G. Corneil, H. Lerchs, L. Stewart Burlingham.
Complement reducible graphs. {\em Discrete Applied Mathematics}, 3: 163--174, 1981.

\bibitem{Ch15}
J.~Chuzhoy.
\newblock Improved bounds for the flat wall theorem.
\newblock {\em Proc. {SODA} 2015}, pages 256--275, 2015.

\bibitem{DJP19}
K.~K. Dabrowski, M. Johnson, and D. Paulusma.
\newblock Clique-width for hereditary graph classes.
\newblock {\em Proc. BCC 2019, London Mathematical Society Lecture Note
  Series}, 456:1--56, 2019.

\bibitem{DH08}
D.~Du and X.~Hu.
\newblock {\em {S}teiner {T}ree {P}roblems in {C}omputer {C}ommunication
  {N}etworks}.
\newblock World Scientific, 2008.

\bibitem{FPW85}
M.~Farber, W.~R. Pulleyblank, and K.~White.
\newblock Steiner trees, connected domination and strongly chordal graphs.
\newblock {\em Networks}, 15:109--124, 1985.

\bibitem{FH77}
S.~F\"oldes and P.~L. Hammer.
\newblock Split graphs.
\newblock {\em Congressus Numerantium}, 19:311--315, 1977.

\bibitem{GJ77}
M.~R. Garey and D.~S. Johnson.
\newblock The rectilinear {S}teiner tree problem is {N}{P}-complete.
\newblock {\em SIAM Journal on Applied Mathematics}, 32:826--834, 1977.

\bibitem{GJPS17}
P.~A. Golovach, M.~Johnson, D.~Paulusma, and J.~Song.
\newblock A survey on the computational complexity of colouring graphs with
  forbidden subgraphs.
\newblock {\em Journal of Graph Theory}, 84:331--363, 2017.

\bibitem{Jo98}
{\"O}. Johansson.
\newblock Clique-decomposition, {NLC}-decomposition, and modular decomposition
  -- relationships and results for random graphs.
\newblock {\em Congressus Numerantium}, 132:39--60, 1998.

\bibitem{Ka72}
R.~M. Karp.
\newblock Reducibility among combinatorial problems.
\newblock {\em Complexity of Computer Computations}, pages 85--103, 1972.

\bibitem{LR}
V. Lozin and I. Razgon.
Tree-width dichotomy.
{\em CoRR}, abs/2012.01115, 2020.

\bibitem{MPS19}
B.~Martin, D.~Paulusma, and S.~Smith.
\newblock Colouring ${H}$-free graphs of bounded diameter.
\newblock {\em Proc. MFCS 2019, LIPIcs}, 138:1--14, 2019.

\bibitem{OS06}
S.-I. Oum and Paul Seymour.
\newblock Approximating clique-width and branch-width.
\newblock {\em Journal of Combinatorial Theory, Series B}, 96:514--528, 2006.

\bibitem{PS12}
H.~J. Pr{\"o}mel and A.~Steger.
\newblock {\em The {S}teiner {T}ree {P}roblem: {A} {T}our through {G}raphs,
  {A}lgorithms, and {C}omplexity}.
\newblock Springer Science \& Business Media, 2012.

\bibitem{Ra08}
M. Rao.
\newblock Clique-width of graphs defined by one-vertex extensions.
\newblock {\em Discrete Mathematics}, 308:6157--6165, 2008.

\bibitem{RS18}
P.~Renjith and N.~Sadagopan.
\newblock The {S}teiner tree in ${K}_{1,r}$-free split graphs - a dichotomy.
\newblock {\em Discrete Applied Mathematics}, 280:246--255, 2020.

\bibitem{RS86}
N.~Robertson and P.~D. Seymour.
\newblock Graph minors. {V}. {E}xcluding a planar graph.
\newblock {\em Journal of Combinatorial Theory, Series B}, 41:92--114, 1986.

\bibitem{Tr89}
K. Truemper.
On the delta-wye reduction for planar graphs. {\em Journal of Graph Theory}, 13:141--148, 1989.

\end{thebibliography}
\end{document}